%% file: main.tex
\documentclass[11pt]{article}
\usepackage{amsmath}
\usepackage{amsthm}
\usepackage{amsfonts}
\usepackage{amssymb}
\usepackage{amscd}
\usepackage{epsfig}
\usepackage{epstopdf}
\usepackage{graphicx}
\usepackage{color}
\usepackage{siunitx}
\usepackage[margin=1in]{geometry}
\usepackage{amsmath,amsfonts,amssymb,amsthm}
\usepackage{enumerate}
\usepackage{bbm}
\usepackage{verbatim}
\usepackage{hyperref,color}
\usepackage[capitalize,nameinlink]{cleveref}
\usepackage[dvipsnames]{xcolor}
\hypersetup{
	colorlinks=true,
	pdfpagemode=UseNone,
    citecolor=OliveGreen,
    linkcolor=NavyBlue,
    urlcolor=Magenta,
	pdfstartview=FitW
}
\usepackage{appendix}

\usepackage{tikz}
\usetikzlibrary{intersections}
\usepackage{pgfplots}
\usetikzlibrary{pgfplots.fillbetween}
\usetikzlibrary{patterns}
\pgfplotsset{compat=1.13}

\usepackage{algorithm}
\usepackage{algorithmic}

\newtheorem{theorem}{Theorem}
\newtheorem{lemma}[theorem]{Lemma}

\newtheorem{definition}[theorem]{Definition}
\newtheorem{corollary}[theorem]{Corollary}

\newtheorem{claim}[theorem]{Claim}

\newtheorem{observation}[theorem]{Observation}

\usepackage{graphicx}
\usepackage{subfigure}
\graphicspath{{plots/}}

\DeclareRobustCommand*{\ora}{\overrightarrow}

\begin{document}
\title{Understanding and Generalizing Monotonic Proximity Graphs for Approximate Nearest Neighbor Search}
\author{Dantong Zhu\\ \small Georgia Tech 
        \and 
        Minjia Zhang\\ \small Microsoft Corporation}
\maketitle

\input{abstract}        
\input{intro}           
\input{related}         
\input{MRNG}            
\input{general_MRNG}    
\input{empirical}       
\input{conflict}        
\input{conclusion}      
\input{acknowledgement} 

{\small
\bibliographystyle{ieee_fullname}
\bibliography{reference}
}

\newpage
\appendix
\input{appendix}

\end{document}

%% file: abstract.tex
\begin{abstract}
Graph-based algorithms have shown great empirical potential for the approximate nearest neighbor (ANN) search problem. Currently, graph-based ANN search algorithms are designed mainly using heuristics, whereas theoretical analysis of such algorithms is quite lacking. In this paper, we study a fundamental model of proximity graphs used in graph-based ANN search, called Monotonic Relative Neighborhood Graph (MRNG), from a theoretical perspective. We use mathematical proofs to explain why proximity graphs that are built based on MRNG tend to have good searching performance. We also run experiments on MRNG and graphs generalizing MRNG to obtain a deeper understanding of the model. Our experiments give guidance on how to approximate and generalize MRNG to build proximity graphs on a large scale. In addition, we discover and study a hidden structure of MRNG called conflicting nodes, and we give theoretical evidence how conflicting nodes could be used to improve ANN search methods that are based on MRNG.
\end{abstract}

%% file: intro.tex
\section{Introduction}

Our work is motivated by the classic $k$ nearest neighbor search ($k$NNS) problem. Given a dataset $S$ of $n$ points in $\mathbb{R}^d$, the $k$NNS finds $k$ elements in $S$ that are the closest to a query point $q \in \mathbb{R}^d$ under some metric, for instance, the Euclidean metric. It is a fundamental problem in computational geometry and has many applications in databases and information retrieval. The brute-force algorithm, which checks the distance between $q$ and every element of $S$, takes $O(nd)$ time and is therefore inefficient to run in practice. This motivates us to consider the approximate nearest neighbor (ANN) search problem which, with preprocessing allowed, achieves higher searching efficiency by sacrificing the perfect accuracy of the output. 

Current ANN search algorithms largely fall into the following four lines: space-partitioning tree methods, locality-sensitive hashing, product quantization, and graph-based methods. Recently, graph-based ANN search algorithms have shown better empirical performance than many other classic ANN search algorithms~\cite{hnsw,spread-out-graph,satellite-graph}. In graph-based ANN search, elements in $S$ are preprocessed to be stored in a \textit{proximity graph} $G$, followed by a searching algorithm on $G$ for the queries, where a \textit{proximity graph} on $S$ is a graph in which the set of nodes is $S$ and each node is only adjacent to nodes that are relatively close to it in $\mathbb{R}^d$.

While graph-based ANN search shows great empirical performance, existing algorithms are designed mainly using heuristics and intuition, whereas theoretical analysis of such algorithms is quite lacking. In this paper, we focus on theoretical analysis of graph-based ANN search. Our contribution is as follows:
\begin{itemize}
    \item We provide novel theoretical guarantee on both accuracy and search efficiency for ANN search on Monotonic Relative Neighborhood Graphs (MRNG), which is the fundamental model for the state-of-the-art Navigating Spreading-out Graph (NSG) \cite{satellite-graph} and Satellite System Graph (SSG) \cite{satellite-graph}.
    
    \item We give a mathematically rigorous definition of a parameter space of proximity graphs generalizing MRNG, which subsumes both NSG and SSG.
    
    \item We study the trade-off between efficiency and accuracy on generalized MRNG by comparing the search accuracy on generalized MRNGs with degree upper bound varied. Our study shows that setting a reasonable degree upper bound for generalized MRNG makes both indexing and search much more efficient while maintaining a great level of search accuracy.
    
    \item We discover a hidden structure in MRNG, which is a new category of nodes called conflicting nodes. We provide theoretical justification on how conflicting nodes can be used to improve the search efficiency by helping MRNG escape from local minimum points.
\end{itemize}

The outline of the rest of the paper is as follows. In Section~\ref{sec:related}, we discuss related work. We will formally define monotonicity and Monotonic Relative Neighborhood Graph (MRNG) in Section~\ref{sec:MRNG}, and we will show why they can be useful for graph-based ANN search using mathematical proofs. As MRNG has very expensive indexing cost in a large scale as \cite{satellite-graph} and \cite{satellite-graph} mentioned, we will discuss ways to generalize MRNG in Sections \ref{sec:generalization}-\ref{sec:conflict} to build proximity graphs for graph-based ANN search on a large scale.
In Section~\ref{sec:generalization}, we will discuss two methods of generalizing MRNG, by setting a degree upper bound and restricting a candidate pool from which the algorithm selects neighbors. In Section~\ref{sec:empirical}, we will explore the degree distribution of MRNG and compare the search accuracy in degree-bounded MRNG with the search accuracy in MRNG. Our result indicates that a degree-bounded MRNG with suitable parameters can achieve very similar search accuracy like the MRNG and is significantly more efficient to be constructed and to run search on. Finally, in Section~\ref{sec:conflict} we will share a hidden dimension on generalizing MRNG about a new category of nodes called conflicting nodes. We will provide theoretical evidence on how conflicting nodes can avoid getting stuck at local minimum nodes, and we will discuss several future directions on implementing conflicting nodes in practice.

\textbf{Notations.} We use the following notations: We use $S$ to denote a dataset, $n$ to denote the number of points in $S$, $d$ to denote the dimension of an Euclidean space, and $q$ to denote a query point. For a positive integer $d$, use $\mathbb{R}^d$ to denote the $d$-dimensional Euclidean space. 
For any $x, y \in \mathbb{R}^d$, use $\delta(x, y)$ to denote the $l_2$ distance between them. For any $x \in \mathbb{R}^d$ and $r > 0$, use $B_r(x)$ to denote the open ball of radius $r$ centered at $x$ in $\mathbb{R}^d$. For any $x, y \in \mathbb{R}^d$, $lune(x, y) = B_{\delta(x, y)}(x) \cap B_{\delta(x, y)}(y)$. In a graph $G$, use $V(G)$ and $E(G)$ to denote the set of nodes and edges of $G$, respectively. For any $e \in E(G)$, use $G \backslash e$ to denote the graph obtained from $G$ by deleting the edge $e$. In a directed graph $G$, if there is an edge from $u$ to $v$ for some $u, v \in V(G)$, use $\overrightarrow{uv}$ to denote it and call $v$ an \textit{out-neighbor} or just a \textit{neighbor} of $u$. In a directed graph $G$, for every node $v \in V(G)$, use $N_{out}(v)$ to denote the out-neighborhood of $v$, i.e. $N_{out}(v) = \{ u \in V(G): \overrightarrow{vu} \in E(G)\}$. When we say the degree  or the out-degree of a node $v$ in a directed graph, we mean the number of out-neighbors of $v$, i.e. $|N_{out}(v)|$. In this work, we also stick with $l_2$ metric as the only metric for ANN search.

%% file: related.tex
\section{Related Work} \label{sec:related}

The literature on approximate nearest neighbor search is vast, and hence, we restrict our attention to the most relevant works here. There have been a lot of studies on ANN indexing using space partitioning based methods, which partition the vector space and index the resulting sub-spaces for fast retrievals, such as KD-Tree~\cite{kd-tree}, R{$^*$}-Tree~\cite{r-star-tree}, and Randomized KD-Tree~\cite{flann}. However, the complexity of these methods becomes not more efficient than a brute-force search as the dimension becomes large (e.g., $>$15)~\cite{worst-case-kdtree}. Therefore, they perform poorly on embedding vectors, which are at least a few tens or even hundreds of dimensions. Prior works have also devoted extensive efforts over locality-sensitive hashing (LSH)~\cite{practical-lsh,near-optimal-lsh}. These methods have solid theoretical foundations and allow us to estimate the search time or the probability of successful search. However, LSH and similar approaches have been designed for large
sparse vectors with hundreds of thousands of dimensions. In contrast, we are interested in search of dense continuous vectors with at most a few hundreds of dimensions (e.g., representations learned by neural networks). For these vectors, graph-based approaches outperform LSH-based methods by a large margin on large-scale datasets~\cite{lopq,ann-experiments-analysis,hnsw}. In a separate line of research, people have studied compressing vectors into shortcodes through product quantization~\cite{product-quantization} and its extensions, such as OPQ~\cite{opq}, Cartesian KMeans~\cite{cartesian-kmeans}, and LOPQ~\cite{lopq}. However, although these methods achieve outstanding memory savings, they are sensitive to quantization errors and can result in poor recall@1 accuracy on large datasets~\cite{link-and-code,ann-experiments-analysis}.


More recently, Malkov and Yashunin found that graphs that satisfy the \emph{Small World} property are good candidates for \emph{best-first search}. They introduce Hierarchical Navigable Small World (HNSW)~\cite{hnsw}, which iteratively builds a hierarchical k-NN graph with randomly inserted long-range links to approximate Delaunay Graph~\cite{delaunay-graph}. For each query, it then performs a walk, which eventually converges to the nearest neighbor in logarithmic complexity. Subsequently, Fu et al. proposed NSG, which approximates \emph{Monotonic Relative Neighbor Graph} (MRNG)~\cite{spread-out-graph} that also involves long-ranged links for enhancing connectivity. 
To the best of our knowledge, 
both HNSW and NSG are considered as the state-of-the-art methods for ANN search~\cite{ann-experiments-analysis,lernaean-hydra} and have been adopted by major players in the industry~\cite{nsw,diskann,faiss-source-code}. 
However, since the approximations in these methods are based on heuristics and lack rigorous theoretical support, it leaves questions on whether more effective graphs can be constructed with theoretically grounded methods.  

%% file: MRNG.tex
\section{Monotonicity and Monotonic Relative Neighborhood Graphs (MRNG)} \label{sec:MRNG}

The accuracy of search output and the time complexity of the search are both important factors for every search algorithm. In the area of graph-based ANN search, every work has been trying to optimize both factors, whereas there exists barely any theoretical model that can be proved mathematically to have both high accuracy and good search time complexity. Thus, it is of our interest to study a graph model that achieves both factors. This then leads to our discussion of monotonicity of graphs and Monotonic Relative Neighborhood Graphs (MRNG).

Monotonicity is a property of graphs introduced by Dearholt et al.~\cite{monotonic-search-network} back in the 1980s, before graph-based ANN search was studied, and it was later used to design proximity graphs for graph-based ANN search to facilitate efficient search. In particular, the state-of-the-art proximity models NSG and SSG~\cite{spread-out-graph,satellite-graph} are both based on monotonicity, which is formally defined as follows.

\begin{definition}[Monotonicity]
In a directed graph $G$ with $V(G) \subseteq \mathbb{R}^d$, let $P$ be a path going through nodes $v_1, v_2, ..., v_l$ in order. For $q \in \mathbb{R}^d$, say $P$ is a monotonic path with respect to $q$ if $\delta(v_i, q) > \delta(v_{i+1}, q)$ for all $i = 1,..., l-1$. If $v_l = q$, call $P$ a monotonic path. If there is a monotonic path going from $p$ to $q$ for every $p, q \in V(G)$, then call $G$ a monotonic graph.
\end{definition}

\subsection{Guaranteed accuracy of monotonic graphs}
Lemma~\ref{lemma_monotonic} remarks how monotonicity guarantees search accuracy: In a monotonic graph $G$, for every $p, q \in V(G)$, $q$ can be found starting from $p$ without getting stuck in any local minimum point via Algorithm~\ref{alg_closer_and_go}, which is a simple greedy search algorithm. The proof of Lemma~\ref{lemma_monotonic} can be found in the appendix.

\begin{algorithm}
\caption{closer-and-go($G, p, q$)}
\textbf{Require:} directed graph $G$ with $V(G) \subseteq \mathbb{R}^d$, starting node $p \in V(G)$, query $q \in \mathbb{R}^d$\\
\textbf{Ensure:} A path $P$ going through $p = v_1$, $v_2$, ..., $v_l$ such that $\delta(v_i, q) > \delta(v_{i+1}, q)$ for all $i = 1, ..., l-1$.

\begin{algorithmic}[1]\label{alg_closer_and_go}
\STATE Set $i := 1$, $v_i := p$
\WHILE{$\exists u \in N_{out}(v)$ such that $\delta(u, q) < \delta(v_i, q)$}
    \STATE $i := i+1$
    \STATE $v_i := u$
\ENDWHILE
\end{algorithmic}
\end{algorithm}

\begin{lemma}\label{lemma_monotonic}
Let $G$ be a monotonic graph with $V(G) \subseteq \mathbb{R}^d$. Let $p, q \in V(G)$ be arbitrary, then Algorithm~\ref{alg_closer_and_go} closer-and-go($G, p, q$) finds \emph{a monotonic path from $p$ to $q$ in $G$}.
\end{lemma}

\subsection{Finding the most efficient monotonic graph}
Note that by the definition of monotonicity, if a graph $G$ is monotonic, so is every supergraph of $G$. Also note that the search complexity in a proximity graph in graph-based ANN search is the sum of the number of neighbors searched at each node visited by the search algorithm, which is bounded by the sum of out-degrees of the nodes on the path returned. Therefore, while maintaining the monotocity of the proximity graph to guarantee search accuracy, we would like to have as few edges in the graph as possible so that the search can be done efficiently as well.

Fu et al. \cite{spread-out-graph} points out that Delaunay graph on any given dataset is monotonic, while it suffers degree explosion when the dimension of the dataset is high. This means that searching on a Delaunay graph has good accuracy but is not efficient in practice. Many graph-based ANN search algorithms, such as NSW and HNSW~\cite{nsw,hnsw}, are built from a Delaunay graph by eliminating edges, while their criteria of which edges to eliminate is quite heuristic: it is unclear whether the resulting graph still consists of redundant edges or, on the other hand, lacks some important edges. To understand monotonic graphs better and to use them better, we would like to have a more rigorous criteria for selecting edges to be in the proximity graph of ANN search.

Fu et al. introduced a theoretical model 
called the Monotonic Relative Neighborhood Graph (MRNG), which is a monotonic graph as proved by Fu et al., and it was used to build the practical state-of-the-art NSG~\cite{spread-out-graph}. Fu et al., however, focused on improving the performance of NSG, which is an approximation of MRNG using heuristics, and did not give much theoretical analysis of MRNG. Here, we will restate the formal definition of MRNG, and we will mathematically prove that this is a well-defined model and is an \emph{edge-minimal monotonic graph} as we desire, i.e. deleting any edge from an MRNG would break the monotonicity.

\begin{definition} \label{def_MRNG}
For a finite set $S \subseteq \mathbb{R}^d$, a directed graph $G$ with $V(G) = S$ is called a monotonic relative neighborhood graph (MRNG) on $S$ if for every $x, y \in S$, $\overrightarrow{xy} \in E(G)$ if and only if $\overrightarrow{xz} \not\in E(G)$ for every $z \in lune(x, y) \cap S$.
\end{definition}

\cite{spread-out-graph} introduced Definition~\ref{def_MRNG} and proved that an MRNG is monotonic. However, it is not even clear whether MRNG is well defined, i.e. if MRNG is a unique mathematical object given any dataset and if it is edge-minimal. In this work, we give our own proof that MRNG is a uniquely defined edge-minimal monotonic graph on any given dataset, as stated in Lemma~\ref{lem_MRNG} whose proof can be found in the appendix.

\begin{lemma} \label{lem_MRNG}
Let $S \subseteq \mathbb{R}^d$. There exists a unique MRNG on $S$, and that this MRNG is an edge-minimal monotonic graph.
\end{lemma}


%% file: general_MRNG.tex
\section{Generalization of MRNG} \label{sec:generalization}

Although MRNG seems to be an ideal proximity graph for ANN search based on our analysis in the previous section, as \cite{spread-out-graph} points out, MRNG is very expensive to be constructed in a large scale. In this section, we discuss two methods to generalize MRNG so that it can be utilized better in a large scale: (1) bounding the degree of the graph, and (2) bounding the candidate pool for edge-selection of each node. We call these generalizations \textit{generalized MRNG}. The ultimate goal of studying generalized MRNG is to find some suitable parameters for the generalization such that the resulting graph is efficient to be constructed and, in the meanwhile, maintains a good navigability as MRNG. We note that the state-of-the-art model NSG is a generalized MRNG.

\begin{algorithm}
\caption{build-graph($S, m, (U_1, ..., U_n)$)}
\textbf{Require:} dataset $S \subseteq \mathbb{R}^d$ of $n$ points, degree upper bound $m$, candidate pool $U_x \subseteq S \backslash \{x\}$ for neighbors of $x$ for every $x \in S$\\
\textbf{Ensure:} Generalized MRNG on $S$ with every node degree bounded by $m$
\begin{algorithmic}[1]\label{alg_build_graph}
\FOR{$x \in S$}
    \STATE Set $N_x := \emptyset$ //$N_x$ is to contain neighbors of $x$
    \STATE Sort nodes in $U_x$ in increasing order of distance to $x$
    \WHILE{$U_x \neq \emptyset$ and $|N_x| < m$}
        \STATE $y := $ first node in $U$
        \STATE $U_x := U_x \backslash \{y\}$
        \IF{$\delta(x, y) < \delta(r, y)$ for every $r \in N_x$}
            \STATE $N_x := N_x \cup \{y\}$
        \ENDIF
    \ENDWHILE
\ENDFOR
\end{algorithmic}
\end{algorithm}

By the definition of MRNG, there is a naive algorithm of constructing it: For every point $x \in S$, let $U_x = S \backslash \{x\}$ and sort points in $U_x$ in increasing order of distances to $x$. Then, for every node $y \in U_x$ from the beginning to the end, if $\delta(x, y) < \delta(r, y)$ for every existing neighbor $r$ of $x$ (meaning no existing neighbor $r$ of $x$ is contained in $lune(x, y)$), then add $y$ as a neighbor of $x$. This algorithm has time complexity $O(n^2(\log n + \Delta))$ where $\Delta$ is the maximumn degree of the resulting MRNG graph.

One factor that makes constructing MRNG inefficient is that there could be nodes in an MRNG with very high degree, making the $\Delta$ component big. This is a common problem in many graph-based ANN search algorithms. We note that almost every existing work simply sets an arbitrary degree bound to the graph for convenience. Another two major factors that slow down the construction of MRNG are that the pairwise distance between every two nodes in $S$ have to be computed, and that for every node $x$, all nodes in $S \backslash \{x\}$ have to be sorted by distance to $x$. Therefore, to speed-up this process, we would like to choose the candidate pool $U_x$ to be some non-trivial subset of $S$, and only choose neighbors for $x$ among $U_x$.

In Algorithm~\ref{alg_build_graph}, if the degree of each node is restricted by $m$ for some $m < n$, it builds a \textit{degree-bounded MRNG} with time complexity $O(n^2 (\log n + m))$. When the candidate pool $U_x$ is replaced by some non-trivial subset of $S \backslash \{x\}$ for each $x \in S$, it has time complexity $O(\sum_{x \in S} |U_x|(\log |U_x| + m))$. If for some constant $l$ we have $|U_x| \leq l$ for all $x \in S$, then the complexity would be at most $O(nl(\log l + m))$. A graph on $S$ that is built by Algorithm~\ref{alg_build_graph} with any parameters $m$ and $U_1, ..., U_n$ is a \textit{generalized MRNG}. We note that NSG is a special generalized MRNG whose candidate pool for each node is selected by a $k$NN graph. We see from the previously discussed complexity of $O(nl(\log l + m))$ that, when we construct a generalized MRNG in large scale, the values of $l$ and $m$ determine the trade-off between the navigability and the indexing complexity of the graph.

%% file: empirical.tex
\section{Behaviors of Generalized MRNG} \label{sec:empirical}

With MRNG and generalized MRNG defined and the complexity of them analyzed, we wonder what MRNG and generalized MRNG really look like. In particular, we would like to know how MRNG and generalized MRNG differ from each other in terms of structural properties and navigability for ANN search, as our ultimate goal is to find some generalized MRNG that maintains nice navigability for ANN search and are more efficient to build and run search algorithm on than MRNG.

We focus on degree-bounded MRNG in this section. We will first study a theoretical degree upper bound of MRNG and empirically observe the degree distribution of MRNG as $n$ and $d$ change. We will also study how setting different degree bounds affects the indexing cost and the searching accuracy in the graph, as almost every existing graph-based ANN search algorithm sets an arbitrary degree upper bound of the graph to reduce both the indexing and searching complexity. We were fairly surprised to discover that setting a reasonable upper bound of the graph in fact does not affect the searching accuracy of MRNG in both high and low dimensions. This gives us more confidence to set a degree bound when generalizing MRNG in practice. 

\subsection{Exponential degree upper bound in theory}
Observe that by the definition of MRNG, the degree between any two edges coming out of the same node is no smaller than $60^\circ$. By an $\epsilon$-net argument in Lemma 5.2 of \cite{epsilon-net}, we can obtain the following lemma.
\begin{lemma} \label{lem_exp_bound}
For $S \subseteq \mathbb{R}^d$, let $G$ be the MRNG on $S$. Then, the maximum degree of $G$ is at most $O((1 + 6/\pi)^d)$.
\end{lemma}

Lemma~\ref{lem_exp_bound} shows that the upper bound of the MRNG of any dataset of $n$ points in $d$-dimension is in fact independent of $n$. This means that there is no degree explosion for fixed dimension as the size of the dataset grows. We note that \cite{spread-out-graph} also proved that the maximum degree of MRNG is independent from $n$, but their proof did not point out that the degree bound could in fact heavily depend on $d$. The exponential dependency on $d$ needs to be mentioned, as in practice the dimension of data could be quite different from one to another, from a hundred to a thousand.

\subsection{Observation of degree distribution of MRNG}

Although the degree upper bound in Lemma~\ref{lem_exp_bound} is independent from the size of the dataset, it does not look as satisfying since it is exponential in the dimension $d$ of the space. However, it is not clear whether or not the bound $O((1 + 6/\pi)^d)$ is tight. For every positive integer $d$, let $T(d)$ be the maximum degree of an MRNG among the MRNGs of all datasets in $\mathbb{R}^d$. Lemma~\ref{lem_exp_bound} tells us that $T(d) = O((1 + 6/\pi)^d)$. But is $T(d)$ equal to $\Theta((1 + 6/\pi)^d)$, i.e. is $T(d)$ really exponential in $d$? To our best knowledge, the answer is unknown.

To better understand the degree bound of MRNG, we took an empirical approach. We constructed MRNG with various $n$ and $d$ and observed the degree distributions empirically. In our experiment, with each given pair of $n$ and $d$, we independently generate $n$ points uniformly at random from the unit hyper-cube $[0, 1]^d$ to build the dataset $S$. We then run Algorithm~\ref{alg_build_graph} with $m = n$ and $U_x = S \backslash \{x\}$ for all $x \in S$ to build MRNG on the dataset generated.

\begin{figure}
\centering
\subfigure[mean = 11, min = 2, max = 27]{\includegraphics[scale = 0.5]{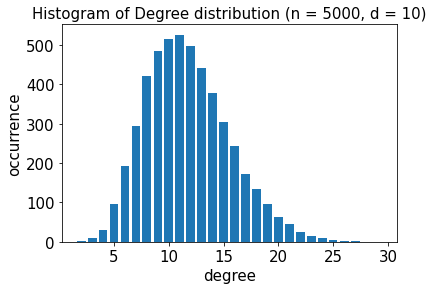}}%
\subfigure[mean = 37, min = 5, max = 203]{\includegraphics[scale = 0.5]{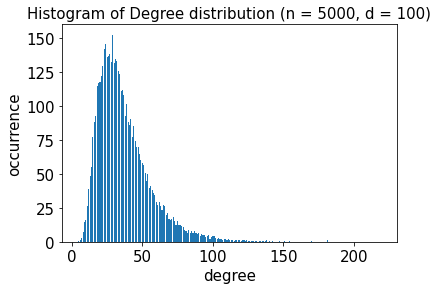}}%
\caption{Histogram of Degree Distribution of MRNG on 5000 nodes, with $d = 10, 100$}
\label{fig_deg_hist}
\end{figure}

In Figure~\ref{fig_deg_hist}, we plot the distributions of the degrees of nodes in an MRNG with $n$=5000, $d$= 10 and another MRNG with $n$=5000, $d$=100. We choose $n$=5000 because the MRNG cannot be efficiently constructed with larger $n$ due to the $O(N^2)$ computation needed to rank near neighbors for each point in the dataset. We observe that on the same number of data points, the average degree in 100-dimension is more than 3 times as big as the average degree in 10-dimension, and the maximum degree in 100-dimension is more than 7.5 times as big as the maximum degree in 10-dimension. Moreover, the degree distribution behaves more like a normal distribution when the dimension is small, whereas it is way more skewed toward to lower degrees when the dimension is high with a few points of high degree. These observations agree with the curse of dimensionality: constructing MRNG in high dimension requires a significantly larger amount of time than constructing MRNG in low dimension. The skewness of the degree distribution in high dimension also suggests that a decent amount of time would be spent on only a few nodes that have high degree during the construction of MRNG in high dimension. These observations then raise the question that how much the extra edges in a high-dimensional MRNG is really helpful with navigation in a graph-based ANN search algorithm? We will discuss this question more deeply in Section 4.3.

\begin{figure}
\centering
\subfigure{\includegraphics[scale=0.5]{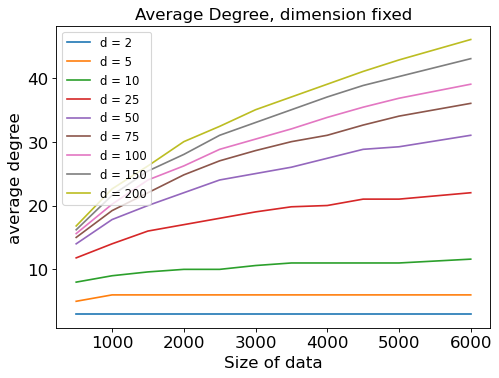}}%
\subfigure{\includegraphics[scale=0.5]{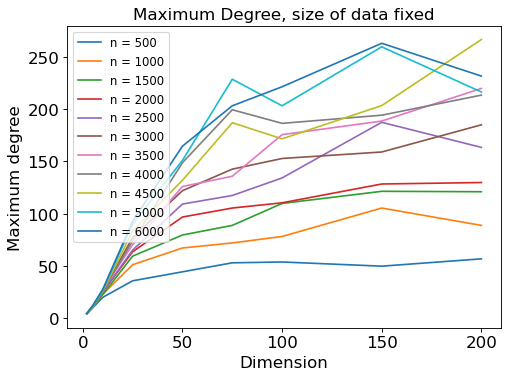}}%
\caption{Average and maximum degrees of MRNG as $n$ and $d$ vary.}
\label{fig_deg_trend_d}
\end{figure}

In Figure~\ref{fig_deg_trend_d}, we observe the trend of average and maximum degrees of MRNG as $n$ and $d$ vary: (1) When $d$ is fixed and is not more than $10$, we observe that both average and maximum degrees of MRNG hit a plateau as $n$ increases. We believe this plateau value for maximum degree is around the theoretical bound $T(d)$. We also observe that the plateau value for average degree remains no more than a third of the plateau value for the maximum degree as $n$ goes up. This seems to imply that there exists an asymptotic upper bound for the average degree of MRNG which is independent from the size of the dataset and maintains certain distance from the maximum degree tight bound $T(d)$. (2) When the dimension is high, we do not observe a plateau of the maximum degree any more where $n$ is set to be no more than 6000. (3) Furthermore, we observe that with the fixed number of data points, both the average degree and the maximum degree increase as $d$ goes up with a decreasing rate.

We believe that $T(d)$, which is the maximum degree of an MRNG among the MRNGs of all datasets in $\mathbb{R}^d$, is as we observed for the plateau values for small dimensions in Figure~\ref{fig_deg_trend_d}. It cannot yet be observed for higher dimensions, as we restricted the size of the data to be no more than 6000. Observation (3) gives us some evidence to suspect that the bound $O((1 + 6/\pi)^d)$ is not tight for $T(d)$, and that it is possible that $T(d)$ is sub-exponential in $d$. This means that constructing an MRNG in high dimension may be more efficient than it seems in Lemma~\ref{lem_exp_bound}.

\subsection{Best-first search on degree-bounded MRNG}
In theory, we showed that every MRNG is an edge-minimal monotonic graph, and therefore any non-trivial degree-bounded MRNG is not monotonic. However, although not monotonic, how navigable is a graph that is obtained from MRNG by eliminating some edges, and how does its navigability differ from the navigability of MRNG? Note that the number of edges in the proximity graph heavily affects the searching time complexity, which is simply the the sum of neighbors searched at each node visited by the search algorithm. If adding certain edges in MRNG does not contribute to the navigability as much, it would make sense to eliminate them so that the search is more time efficient while maintaining the accuracy of the search. In almost every practical ANN search proximity model, an arbitrary degree upper bound is set. However, a fair comparison among different degree bounds is missing. In this section, we explore how different values of degree bounds affect the searching efficiency in degree-bounded MRNGs.



In our experiment, for a given pair of $n$ and $d$, we again independently generate $n$ points uniformly at random from the unit hyper-cube $[0, 1]^d$ to build the dataset $S$. We then run Algorithm~\ref{alg_build_graph} with varied values of $m$ and $U_x = S \backslash \{x\}$ for all $x \in S$ to build the degree-bounded MRNGs. With 200 queries that are generated from $[0, 1]^d$ independently uniformly at random, we use best-first search 
to search for the top 1 nearest neighbor of each query and use the averaged accuracy among these 200 queries as a measurement of navigability in our analysis. We note that best-first search is a commonly used searching algorithm in many existing graph-based ANN search works~\cite{hnsw,spread-out-graph,satellite-graph}. In a best-first search, we have the number of nodes checked as a parameter, which allows us to control the total number of computation distances easily. Note that if the number of nodes checked is equal to $n$, the algorithm checks every node in the graph and therefore would guarantee to find the true nearest neighbor.

\begin{figure}
\centering
\subfigure[average degree = 21, maximum degree = 90]{\includegraphics[scale = 0.50]{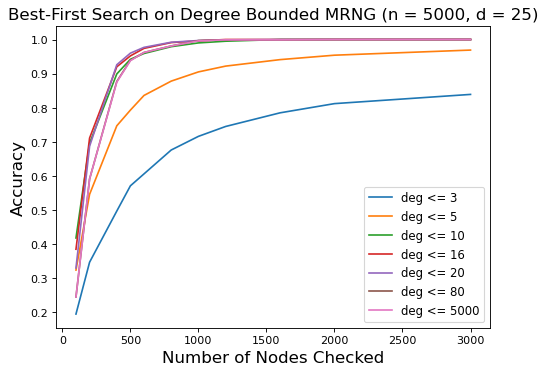}}%
\subfigure[average degree = 37, maximum degree = 203]{\includegraphics[scale = 0.50]{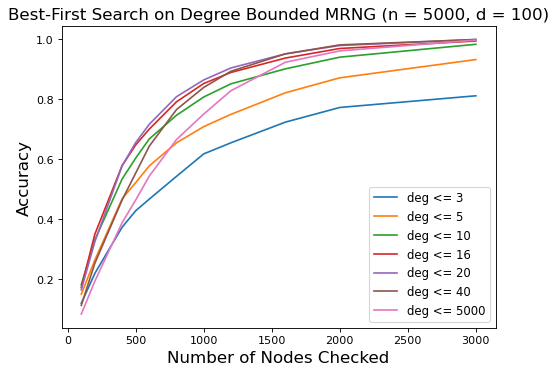}}%
\caption{Best-first search accuracy on degree-bounded MRNG}
\label{figure_bounded_beam}
\end{figure}

Figure~\ref{figure_bounded_beam}(a) shows search accuracy on degree-bounded MRNG with varied defgree bounds for $n$=5000, $d$=25 and plot (b) for $n$=5000, $d$=100. We observe that for both low and high dimensions, best-first search on some degree-bounded MRNG works as well as on the authentic MRNG. In particular, when $d$=25, MRNG has average degree of 21 and maximum degree 90: the accuracy curve for degree bound 10 matches the accuracy for unbounded degree, both achieving around 0.95 accuracy after checking 500 nodes. When $d$=100, MRNG has average degree 37 and maximum degree 203: the accuracy curve for degree bound 18 matches the accuracy for unbounded degree, both achieving around 0.90 accuracy after checking 1200 nodes. We also observe that accuracy seems to be the highest when the degree bound is set to 20, and that when the degree bound is larger than 20 the accuracy starts to fall down.

Recall that when constructing MRNG in Algorithm~\ref{alg_build_graph}, the neighbors of every node $x$ in an MRNG are added in the order of increasing order of distance to $x$. When setting the degree upper bound to $m$, the neighbors of $x$ simply become the top $m$ neighbors of $v$ in the exact MRNG. Therefore, our above observations indicate that adding neighbors of $x$ in MRNG that are further away from $x$ does not really make the search more effective. We conclude that setting an approximate upper bound on the degree in a generalized MRNG not only improves the searching time complexity on the graph but also helps with the search accuracy. In the experiment with $n = 5000$ and $d = 25$ or $100$, we see the optimal degree bound is only about one half of the average degree of the exact MRNG. This means that by setting an appropriate degree bound, besides improvements on search time complexity and accuracy, we can also save at least a half of the indexing cost comparing to constructing the original MRNG.

Furthermore, we observe a threshold of the degree upper bound which makes a phase transition of the search accuracy for both low and high dimensions. In particular, we see that the search accuracy can be significantly improved when the degree upper bound is set to below 10 or 16 for low and high dimensions, respectively, and once the degree upper bound is greater than this threshold, there is no significant change of the search accuracy. We believe that this phase transition of search accuracy indicates a phase transition of certain properties of the graph. For instance, the connectivity of the graph. When the degree upper bound is very small, the graph is a union of many disconnected small subgraphs; as the bound grows, these subgraphs start to merge, the number of connected components goes down, and the search accuracy goes up; once the bound reaches the threshold, the graph achieves a certain level of connectivity, and any stronger version of connectivity seems to not help with the navigability for finding nearest neighbors in the graph any more. The phase transition phenomenon could be related to other properties of the graph as well.

%% file: conflict.tex
\section{Conflicting Nodes: Generalizing MRNG to Escape From Local Minimum} \label{sec:conflict}

From a mathematical point of view, the edge-selection criteria of MRNG is natural and pretty as a definition for a fundamental graph model. Besides studying degree upper bound and restricted candidate pools for selecting neighbors for each node, we believe that there exist deeper structural properties of MRNG that are worth exploring to be used to generalize MRNG, which may eventually be useful to improve current graph-based ANN search work in large scale.

In this section, we discuss a new category of nodes called \textit{conflicting nodes}. It is a hidden structure of MRNG, which comes up very naturally during the construction process of MRNG. We will provide theoretical evidence why conflicting nodes can be useful for getting out of local minimum points in a graph-based ANN search algorithm on MRNG and generalized MRNG so that it can eventually improve the accuracy and time complexity of the search. We will also discuss some challenges and future directions about implementing conflicting nodes in practice. To the best of our knowledge, conflicting nodes have not been mentioned or utilized by any other ANN search related work so far.

Intuitively, a conflicting node $w$ of an edge $\overrightarrow{vu}$ in an MRNG is a node that failed to become a neighbor of $v$ due to $u$. When $v$ is a local minimum point with respect to some query, every neighbor $u$ of it is no closer to the query than $v$, so it would then make sense to search for true nearest neighbors among the nodes that did not get to become a neighbor of $v$ and are still relatively close to $v$. Figure~\ref{fig:conf_demo} shows how a conflicting node can be helpful to get around a local minimum node $v$ when projecting to a 2-dimensional space: Suppose $u, w \in S$ and $u$ is the closest to $v$ in $S$, so $\overrightarrow{vu}$ is an edge in the MRNG on $S$. Observe that $u \in lune(v, w)$, and thus $\overrightarrow{vw}$ is not an edge of the MRNG on $S$, meaning that a normal greedy search algorithm would not check $w$ when getting to $v$. We observe that $\delta(w, q) < \delta(v, q)$ in Figure~\ref{fig:conf_demo}, and we know that $w$ can be related to $v$ via the neighbor $u$ of $v$. Therefore, if we can record extra information of conflicting nodes like $w$ and their relationships to the neighbors of $v$, we can then have the
search algorithm to skip to $w$ when being stuck at $v$.

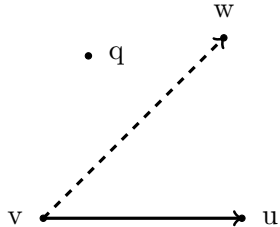
\begin{figure}[th]
    \centering
    \begin{tikzpicture} [scale = 1.2]
    \node[label=left:v](v) at (0,0){};
    \node[label=right:u](u) at (2.2,0){};
    \node[label=above:w](w) at (2,2){};
    \node[label=right:q](q) at (0.5, 1.8){};
    
    \filldraw (v) circle (1pt);
    \filldraw (u) circle (1pt);
    \filldraw (w) circle (1pt);
    \filldraw (q) circle (1pt);

    \draw[very thick, ->] (v.center) -- (u.center);
    \draw[dashed, very thick, ->] (v.center) -- (w.center);
    \end{tikzpicture}
    
    \caption{$q$ is the query; $v \in S$ is a local minimum node with respect to $q$; $u \in S$ is the nearest neighbor of $v$ in $S$; $w$ is a conflicting node to $\ora{vu}$ that is closer to $q$ than $v$.}
    \label{fig:conf_demo}
\end{figure}

Formally, we define local minimum nodes and conflicting nodes as follows.

\begin{definition}
In a directed graph $G$ with $V(G) \subseteq \mathbb{R}^d$, for any $v \in V(G)$ and a point $q \in \mathbb{R}^d$, $v$ is called a local minimum node with respect to $q$ and $G$ if $\delta(u, q) \geq \delta(v, q)$ for every $u \in N_{out}(v)$.
\end{definition}

\begin{definition}
For $S \subseteq \mathbb{R}^d$, let $G$ be the MRNG on $S$, and let $v, u \in S$ such that $\overrightarrow{vu} \in E(G)$. A node $w \in S$ is called a conflicting node of $\overrightarrow{vu}$ if $u \in lune(v, w)$. We use $C(\overrightarrow{vu})$ to denote the set of all conflicting nodes of $\overrightarrow{vu}$, i.e. $C(\overrightarrow{vu}) = \{w \in S: u \in lune(v, w)\}$.
\end{definition}

Due to the nature of an MRNG, we can simply derive the following observation, which says that when $v$ is a local minimum point with respect to a query $q$, if $v$ is not the true nearest neighbor for $q$, then the true nearest neighbor for $q$ must be a conflicting node of some edge on $v$.

\begin{observation}  \label{obs_conf_abs_NN}
Let $S \subseteq \mathbb{R}^d$ and $G$ be the MRNG on $S$. Let $v \in S$ and $q \in \mathbb{R}^d$ such that $v$ is a local minimum node with respect to $q$ and $G$. Let $w = \arg \min_{x \in S} \delta(x, q)$. If $w \neq v$, then $w \in C(\overrightarrow{vu})$ for some $u \in N_{out}(v)$.
\end{observation}
Note that for every $x \in S$, $\bigcup_{y \in N_{out}(x)} C(\overrightarrow{xy}) = S - N_{out}(x) \cup \{x\}$ is the set of all non-neighbors of $x$, and it is not practical to search for all of them to find the true nearest neighbor. It turns out that given a query $q$ and a local minimum node $v$ with respect to $q$, with any data distribution, it is just impossible to have the true nearest neighbor for $q$ to be a conflict node of some neighbor of $v$. Figure~\ref{figure_adv_region} shows a good example for this in 2-dimension, in which $q$ is the query, $v$ is a local minimum point with respect to $q$, and $w$ is the true nearest neighbor of $q$ in $S$. Since $v$ is a local minimum point with respect to $q$, $v$ has no neighbor in the open ball centered at $q$ of radius $\delta(v, q)$ (within the red circle in Figure~\ref{figure_adv_region}). However, since $\overrightarrow{vw}$ is not an edge in the MRNG, there must exist some $u \in lune(v, w)$ that is a neighbor of $v$. Therefore, such a neighbor $u$ must be contained in $lune(v, w)$ but outside the open ball centered at $q$ of of radius $\delta(v, q)$ (red circle), precisely the shaded region in Figure~\ref{figure_adv_region}. This means that if there exists a node in $S$ at the position of $w$, then $v$ must have some neighbor in the shaded region.

\begin{figure}[th]
    \centering
    \begin{tikzpicture} [scale = 1.5]
        \draw[thin, ->,>=latex] (-1,0)--(2.3,0) node[above] {};
        \draw[thin, ->,>=latex] (0,-1)--(0,1.7) node[left] {};
        \node [label={[xshift=0.5cm, yshift=-0.7cm]v(0, 0)}] (v) at (0,0){};
        \node [label={[xshift=0.2cm, yshift=-0.7cm]q(1, 0)}] (q) at (1,0){};
        \node [label={[xshift=0.4cm, yshift=-0.7cm](2, 0)}] (v2) at (2,0){};
        \node [label={[xshift=0.2cm, yshift=-0.3cm]u}] (u) at (0.2,1){};
        \node [label={[xshift=0.2cm, yshift=-0.3cm]w}] (w) at (1.079,0.784){};
        \filldraw (v) circle (1pt);
        \filldraw (q) circle (1pt);
        \filldraw (v2) circle (1pt);
        \filldraw (u) circle (1pt);
        \filldraw (w) circle (1pt);
        \draw[red, domain=0:360,samples=500] plot (\x:{2*cos(\x)});
        \draw[dotted, domain=-45:120,samples=500] plot (\x:{4/3});
        \draw[dotted, domain=90:165,samples=500] plot (\x:{(8/3)*cos(\x - 36)});
        \draw[thick, ->] (v.center) -- (u.center);
        \draw[thick, dashed, ->] (v.center) -- (w.center);
        
        \draw[dotted, domain=47:90,samples=500, name path=A1] plot (\x:{4/3});
        \draw[red, domain=47:90,samples=500, name path=B1] plot (\x:{2*cos(\x)});
        \tikzfillbetween[of=A1 and B1]{gray, opacity=0.2};
        
        \draw[dotted, domain=90:96,samples=500, name path=A2] plot
        (\x:{4/3});
        \draw[white, opacity=0, name path=B2] (-0.139, 1.326) -- (0, 1.326) ;
        \tikzfillbetween[of=A2 and B2]{gray, opacity=0.2};
        
        \draw[dotted, domain=96:120,samples=500, name path=A3] plot (\x:{(8/3)*cos(\x - 36)});
        \draw[thin, name path=yaxis] (0,0)--(0,4/3);
        \tikzfillbetween[of=A3 and yaxis]{gray, opacity=0.2};
    \end{tikzpicture}
    \caption{$q$ is the query, $v$ is a local minimum point with respect to $q$. If a true nearest neighbor of $q$ in $S$ is at the position of $w$, then $v$ must have some neighbor $u$ contained in $lune(v, w)$ but outside the open ball centered at $q$ of of radius $\delta(v, q)$, precisely the shaded region.}
\label{figure_adv_region}
\end{figure}
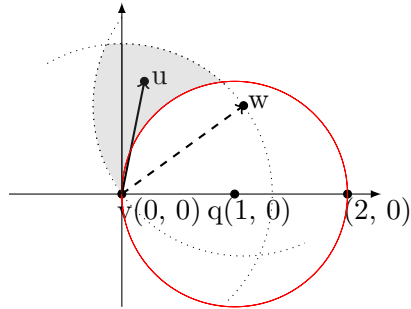

Figure~\ref{figure_adv_region} shows that to check if there exists a node at the position of $w$ in $S$, it suffices to check if $v$ has any neighbor in the shaded region. Let this region be $R_w$. Since every node in $S$ closer to $q$ than $v$ must be contained in the ball centered at $q$ of radius $\delta(v, q)$, it follows that if $v$ is not the global minimum of $q$, then $v$ must have some neighbor in the union of regions $R_w$ over all $w \in B_{\delta(v, q)}(q)$.

In the next lemma, we will give the precise form of this union of regions, and
 we will generalize this result to dimension $d$ for any arbitrary $d \geq 2$.

\begin{lemma}\label{lem_search_critiria}
Let $S \subseteq \mathbb{R}^d$ and $G$ be the MRNG on $S$. Let $v \in S$ and $q \in \mathbb{R}^d$ such that $v$ is a local minimum node with respect to $q$ in $G$. Let $u \in N_{out}(v)$ such that $u \not\in B_{\delta(v, q)}(q)$. Let $d = \delta(v, u)$ and $\theta = \angle qvu \in [0, \pi]$. Then, there exists some $w \in B_{\delta(v, q)}(q)$ such that $u \in lune(v, w)$ if and only if $d < \delta(v, q) \cdot f(\theta)$ where
\begin{equation*}
f(\theta) = 
\begin{cases}
2,  & \text{if  } \theta \in [0, \frac{1}{3}\pi]\\
2 \cos(\theta - \pi/3), & \text{if } \theta \in (\frac{1}{3}\pi, \frac{2}{3}\pi]\\
2(\cos \theta + 1), & \text{if } \theta \in (\frac{2}{3}\pi, \pi].
\end{cases}
\end{equation*}
\end{lemma}

We can visualize Lemma~\ref{lem_search_critiria} in 2-dimension in Figure~\ref{figure_heart}: Without loss of generality, we place the query $q$ at $(1, 0)$ and the local minimum point $v$ with respect to $q$ at the origin. Since $v$ is a local minimum point with respect to $q$, $v$ does not have any neighbor inside the unit ball centered at $q$, which is region inside the circle in red in Figure~\ref{figure_heart}. If $v$ is not the global minimum, by Lemma~\ref{lem_search_critiria}, the global minimum $w$ must be a conflicting node to some $\overrightarrow{vu}$ where $u$ is contained inside the blue closed curve but outside the red circle, the shaded region in gray. In other words, to search for nodes closer to $q$ then $v$, Lemma~\ref{lem_search_critiria} says that it suffices to check conflicting nodes to edges $\overrightarrow{vu'}$ for $u'$ contained in the gray shaded region.
The proof of Lemma~\ref{lem_search_critiria} is quite technical, and we leave it to the appendix of this paper. Now with Lemma~\ref{lem_search_critiria}, we can then derive the following Corollary~\ref{cor_heart}.

\begin{figure}[th]
\centering
\begin{tikzpicture}[scale = 1.1]
    \draw[thin, ->,>=latex] (-1,0)--(2.5,0) node[above] {};
    \draw[thin, ->,>=latex] (0,-2)--(0,2) node[left] {};
    \node [label={[xshift=0.3cm, yshift=-0.7cm]v(0, 0)}] (v) at (0,0){};
    \node [label={[xshift=0.3cm, yshift=-0.7cm]q(1, 0)}] (q) at (1,0){};
    \node [label={[xshift=0.4cm, yshift=-0.7cm](2, 0)}] (u) at (2,0){};
    \filldraw (v) circle (1pt);
    \filldraw (q) circle (1pt);
    \filldraw (u) circle (1pt);
  
    \draw[thick, red, domain=0:360,samples=500] plot (\x:{2*cos(\x)});
    \draw[thick, blue, domain=0:60,samples=500, name path=A1] plot (\x:{2});                           
    \draw[thick, blue, domain=60:90,samples=500, name path=A2_right] plot (\x:{2*cos(\x - 60)});  
    \draw[thick, blue, domain=90:120,samples=500, name path=A2_left] plot (\x:{2*cos(\x - 60)});  
    \draw[thick, blue, domain=120:180,samples=500, name path=A3] plot (\x:{2*(cos(\x)+1)});   
    \draw[thick, red, domain=0:60,samples=500, name path=B1] plot (\x:{2*cos(\x)});
    \draw[thick, red, domain=60:90,samples=500, name path=B2_right] plot
(\x:{2*cos(\x)});
    \draw[thin, name path=yaxis] (0,0)--(0,1.732);
    \tikzfillbetween[of=A1 and B1]{gray, opacity=0.2};
    \tikzfillbetween[of=A2_right and B2_right]{gray, opacity=0.2};
    \tikzfillbetween[of=A2_left and yaxis]{gray, opacity=0.2};
    \tikzfillbetween[of=A3 and yaxis]{gray, opacity=0.2};
    
    \begin{scope}[yscale=-1,xscale=1]
        \draw[thick, blue, domain=0:60,samples=500, name path=A1] plot (\x:{2});               
        \draw[thick, blue, domain=60:90,samples=500, name path=A2_right] plot (\x:{2*cos(\x - 60)});  
        \draw[thick, blue, domain=90:120,samples=500, name path=A2_left] plot (\x:{2*cos(\x - 60)});  
        \draw[thick, blue, domain=120:180,samples=500, name path=A3] plot (\x:{2*(cos(\x)+1)});   
        \draw[thick, red, domain=0:60,samples=500, name path=B1] plot (\x:{2*cos(\x)});
        \draw[thick, red, domain=60:90,samples=500, name path=B2_right] plot (\x:{2*cos(\x)});
        \draw[thin, name path=yaxis] (0,0)--(0,1.732);
        \tikzfillbetween[of=A1 and B1]{gray, opacity=0.2};
        \tikzfillbetween[of=A2_right and B2_right]{gray, opacity=0.2};
        \tikzfillbetween[of=A2_left and yaxis]{gray, opacity=0.2};
        \tikzfillbetween[of=A3 and yaxis]{gray, opacity=0.2};
    \end{scope}
\end{tikzpicture}
\caption{Demonstration of Lemma~\ref{lem_search_critiria} in 2-dimension}
\label{figure_heart}
\end{figure}
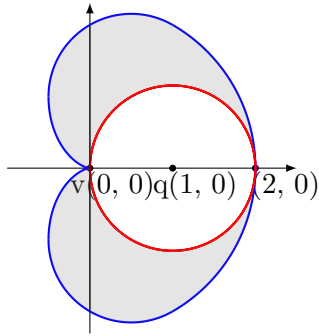

\begin{corollary}[Corollary of Lemma~\ref{lem_search_critiria}] \label{cor_heart}
Let $S \subseteq \mathbb{R}^d$ and $G$ be the MRNG on $S$. Let $v \in S$ and $q \in \mathbb{R}^d$ such that $v$ is a local minimum node with respect to $q$ in $G$. Let $u \in N_{out}(v)$. If $\delta(v, u) \geq \delta(v, q) \cdot f(\theta)$ where $f$ is as defined in Lemma~\ref{lem_search_critiria}, then $\delta(w, q) \geq \delta(v, q)$ for every $w \in C(\overrightarrow{vu})$.
\end{corollary}

We can also make the following simple observation by triangle inequality.
\begin{observation} \label{obs_w_heart}
Let $S \subseteq \mathbb{R}^d$ and $G$ be the MRNG on $S$. Let $v \in S$ and $q \in \mathbb{R}^d$ such that $v$ is a local minimum node with respect to $q$ in $G$. If there exists some $w \in S$ such that $\delta(w, q) < \delta(v, q)$, then $\delta(v, w) < \delta(v, q) + \delta(w, q) < 2 \delta(v, q)$.
\end{observation}

\begin{algorithm}
\caption{conflict-search($S, G, v, q$)}
\textbf{Require:} dataset $S \subseteq \mathbb{R}^d$, MRNG $G$ on $S$ with $C(\overrightarrow{xy})$ for every $\overrightarrow{xy} \in E(G)$, query $q \in V(G)$, local minimum node $v \in S$ with respect to $q$ in $G$\\
\textbf{Ensure:} true nearest neighbor of $q$ in $S$

\begin{algorithmic}[1]\label{Alg_conf_search}
\STATE Set $x:= v$
\STATE Compute $r := \delta(v, q)$
\FOR{$u \in N_{out}(v)$ such that $\delta(v, u) < r \cdot f(\theta)$ where $\theta = \angle qvu$}
    \FOR{$w \in D(\overrightarrow{vu})$ such that $\delta(v, w) < 2 r$}
        \IF{$\delta(w, q) < \delta(v, q)$}
            \STATE $x := w$
        \ENDIF
    \ENDFOR
\ENDFOR
\STATE \textbf{return} $x$
\end{algorithmic}
\end{algorithm}

By Corollary~\ref{cor_heart} and Observation~\ref{obs_w_heart}, it follows that at a local minimum node $v$ with respect to $q$, to find the true nearest neighbor for $q$, it is enough to check nodes $w$ that is a conflicting node for some $\overrightarrow{vu}$ such that $\delta(v, u) < \delta(v, q) \cdot f(\theta)$ and $\delta(v, w) < 2 \delta(v, q)$ where $\theta = \angle qvu$ and $f$ is as defined in Lemma~\ref{lem_search_critiria}. Therefore, Algorithm~\ref{Alg_conf_search} finds a true nearest neighbor of $q$ for every given $q$.

Although conflicting nodes seem to be a great tool to \emph{get out of local minimum points}, we have to mention that it is not so clear how to utilize them to run the search efficiently in practice. We believe that there are many interesting future directions on implementing conflicting nodes efficiently such that the ANN search on generalized MRNG can be improved. We discuss these future directions with more details in the appendix.

%% file: conclusion.tex
\section{Conclusion}

We proved that monotonicity gives great navigability for graph-based ANN search and showed that Monotonic Relative Neighborhood Graphs (MRNG) are well-defined edge-minimal monotonic graphs, using mathematical proofs. We formalized two methods of generalizing MRNG, by setting a degree upper bound and restricting a candidate pool from which the algorithm selects neighbors for each node. These generalizations define a parameter space for proximity graphs for ANN search rigorously, and graphs in this parameter space can be used for ANN search in a large scale. We also explored the degree distribution of MRNG and compared the search accuracy in degree-bounded MRNG with the search accuracy in MRNG. We found out that a degree-bounded MRNG with suitable parameters, although not perfectly monotonic, achieves very similar search accuracy like MRNG and is significantly more efficient to be constructed and to run search on. Finally, we shared our thoughts on another perspective of generalizing MRNG, which is to use conflicting nodes. We proved in theory that conflicting nodes can be used to get around local minimum points when searching on an MRNG and mentioned several future directions on implementing conflicting nodes in practice.

%% file: acknowledgement.tex
\section{Acknowledgement}
The research presented in this paper was partially conducted during the visit of the first author to Microsoft, supported by the NSF supplemental funding of Prof. Robin Thomas at Georgia Tech. The first author thanks Yuxiong He and Junhua Wang at Microsoft for their kind hospitality and Robin Thomas for his generous support. Both authors thank Niranjan Uma Naresh for helpful discussions.

%% file: appendix.tex
\section{Appendix}
\subsection{Proof of Lemma~\ref{lemma_monotonic}}
\begin{proof}[Proof of Lemma~\ref{lemma_monotonic}]
In every iteration in the while loop, if $q \not\in N_{out}(v_i)$, then by the monotonicity of $G$ we know that there exists a monotonic path going from $v_i$ to $q$. It follows that there exists some $v_{i+1} \in N_{out}(v_i)$ such that $\delta(v_{i+1}, q) < \delta(v_i, q)$. Therefore, the while loop does not terminate until $v_i = q$ for some $i$, and we know that it must terminate since there are only finitely many nodes in $G$.
\end{proof}

\subsection{Proof of Lemma~\ref{lem_MRNG}}
To prove Lemma~\ref{lem_MRNG} that MRNG is a uniquely defined edge-minimal monotonic graph on any given dataset, it suffices to prove Lemmas~\ref{lem_MRNG_mono}, \ref{lem_MRNG_unique}, and \ref{lem_MRNG_min}.

\begin{lemma} \label{lem_MRNG_mono}
Let $S \subseteq \mathbb{R}^d$ and let $G$ be an MRNG on $S$. For every $p, q \in S$, Algorithm~\ref{alg_closer_and_go} closer-and-go($G, p, q$) finds a monotonic path from $p$ to $q$. Therefore, $G$ is a monotonic graph.
\end{lemma}
\begin{proof}[Proof of Lemma~\ref{lem_MRNG_mono}]
Suppose in Algorithm~\ref{alg_closer_and_go}, closer-and-go($G, p, q$) returns a path going through $p = v_1$, $v_2, ..., v_l$. It then suffices to prove that $v_l = q$. For the sake of a contradiction, assume $v_l \neq q$. Then since the algorithm terminates at $v_l$, we know $q \not\in N_{out}(v_l)$. Since $G$ is an MRNG, it follows that there exists some $u \in lune(v_l, q)$ such that $\overrightarrow {v_l u} \in E(G)$. It follows that $\delta(u, q) < \delta(v_l, q)$, a contradiction to the fact that the algorithm terminates at $v_l$.
\end{proof}

\begin{lemma} \label{lem_MRNG_unique}
Let $S \subseteq \mathbb{R}^d$. Then, there is a unique MRNG on $S$.
\end{lemma}
\begin{proof}[Proof of Lemma~\ref{lem_MRNG_unique}]
For the sake of a contradiction, assume that $G_1, G_2$ are both MRNGs on $S$ and that $E(G_1) \neq E(G_2)$. Without loss of generality, there exists some $\overrightarrow{xy} \in E(G_1) - E(G_2)$ such that $\{e \in E(G_1): |e| < \delta(x, y)\} = \{e \in E(G_2): |e| < \delta(x, y)\}$. Let $A_i = \{z \in S: \overrightarrow{xz} \in E(G_i), \delta(x, z) < \delta(x, y)\}$ for $i = 1, 2$. It follows that $A_1 = A_2$. Then, by Definition~\ref{def_MRNG}, $\overrightarrow{xy} \in E(G_i)$ if and only if $A_i = \emptyset$, meaning that $\overrightarrow{xy} \in E(G_1)$ if and only if $\overrightarrow{xy} \in E(G_2)$, a contradiction to the assumption that $\overrightarrow{xy} \in E(G_1) - E(G_2)$.
\end{proof}

\begin{lemma} \label{lem_MRNG_min}
Let $S \subseteq \mathbb{R}^d$ and let $G$ be the MRNG on $S$. Then, $G$ is an edge-minimal monotonic graph, i.e. for every $e \in E(G)$, $G \backslash e$ is not monotonic.
\end{lemma}

\begin{proof}[Proof of Lemma~\ref{lem_MRNG_min}]
Let $e \in E(G)$, and say $e = \overrightarrow{xy}$ for $x, y \in S$. Let $G' = G \backslash e$. It suffices to show that $G'$ is not monotonic. For the sake of a contradiction, assume $G'$ is monotonic, and therefore there is a monotonic path $P$ from $x$ to $y$ in $G'$ that goes through $x = v_1, v_2, ..., v_l = y$ in order. Since $\overrightarrow{xy} \not\in E(G')$, we know $l \geq 3$ and $x, v_2, y$ are three distinct vertices. It follows that $\delta(x, y) > \delta(v_2, y)$ and therefore $v_2 \in B_{\delta(x, y)}(y)$. Since $G$ is an MRNG on $S$, $\overrightarrow{xy}
\in E(G)$ implies that there is no $z \in lune(x, y) \cap S$ such that $\overrightarrow{xz} \in E(G)$. Since $\overrightarrow{xv_2} \in E(G)$, we know that $v_2 \in B_{\delta(x, y)}(y) - lune(x, y)$, and therefore $\delta(x, v_2) > \delta(x, y)$. Since $\delta(x, y) > \delta(v_2, y)$, we have $\delta(x, v_2) > \max \{\delta(x, y), \delta(v_2, y)\}$, meaning that $y \in lune(x, v_2) \cap S$. This then implies that $\overrightarrow{x v_2}$ should not be an edge in $G$ due to the MRNG edge selection critiria, a contradiction.
\end{proof}

\subsection{Proof of Lemma~\ref{lem_search_critiria}}
We now give our complete proof of Lemma~\ref{lem_search_critiria}, which is restated as follows.\\
\textbf{Lemma~\ref{lem_search_critiria}.} \textit{Let $S \subseteq \mathbb{R}^d$ and $G$ be the MRNG on $S$. Let $v, u \in S$ and $q \in \mathbb{R}^d$ such that $v$ is a local minimum node with respect to $q$ in $G$ and $u \in N_{out}(v)$. Let $r = \delta(v, q)$, $d = \delta(v, u)$, and without loss of generality let $\theta = \angle qvu \in [0, \pi]$. Then, there exists some $w \in B_r(q)$ such that $u \in lune(v, w)$ if and only if $d < r \cdot f(\theta)$ where
\begin{equation*}
f(\theta) = 
\begin{cases}
2,  & \text{if  } \theta \in [0, \frac{1}{3}\pi]\\
2 \cos(\theta - \pi/3), & \text{if } \theta \in (\frac{1}{3}\pi, \frac{2}{3}\pi]\\
2(\cos \theta + 1), & \text{if } \theta \in (\frac{2}{3}\pi, \pi].
\end{cases}
\end{equation*}}

\begin{proof} [Proof of Lemma~\ref{lem_search_critiria}]
For convenience, let $B = B_r(q)$, and we use $\partial(B)$ to denote the boundary of $B$. Note that either $u, v, q$ are collinear or not. In both cases, we can find a 2-dimensional plane that contains all three of them. Let $P$ denote this plane.

The outline of the proof of Lemma~\ref{lem_search_critiria} is as follows. We will first prove Lemma~\ref{lem_redu_dim} which reduces the problem from $d$-dimension ($d \geq 2$ is arbitrary) to the boundary of $B$ in 2-dimension by proving that there exists $w \in B$ such that $u \in lune(v, w)$ if and only if there exists some $w' \in P \cap \partial(B)$ such that $u \in lune(v, w')$. We will then prove Lemma~\ref{lem_angle} and Lemma~\ref{lem_bound}, which link the reduced problem in 2-dimension to the bound $r \cdot f(\theta)$.

\begin{lemma} \label{lem_redu_dim}
 There exists some $w \in B$ such that $u \in lune(v, w)$ if and only if there exists some $w' \in P \cap \partial(B)$ such that $u \in lune(v, w')$.
\end{lemma}
\begin{proof}[Proof of Lemma~\ref{lem_redu_dim}]
$(\Leftarrow)$: This is the easy direction. Assume that there exists $w' \in P \cap \partial(B)$ such that $u \in lune(v, w')$. This means that $\delta(v, w') > \max \{\delta(v, u), \delta(w', u)\}$. Let $\epsilon = \frac{1}{2} (\delta(v, w') - \max \{\delta(v, u), \delta(w', u)\})$, and note that $\epsilon < \delta(v, w')$ by its definition. Let $w$ be the point on the line segment $v w'$ such that $\delta(w, w') = \epsilon$. It follows that $w \in B$ and $\delta(v, w) > \max \{\delta(v, u), \delta(w', u)\}$, meaning $u \in lune(v, w)$. Therefore, $w$ is as desired.

$(\Rightarrow)$: Let $w_0$ be the unique projection point of $w$ onto $P$, and let $d_0 = \delta(w, w_0)$. Let $l_{v w_0}$ be the line going through $v$ and $w_0$, and let $w'$ be the unique intersection point of $l_{v w_0}$ and $P \cap B$ that is not equal to $v$. It now suffices to prove that $u \in lune(v, w')$, which requires both $\delta(u, w') < \delta(v, w')$ and $\delta(u, v) < \delta(v, w')$

We first prove that $\delta(u, w') < \delta(v, w')$. Note that $\delta(v, w_0)^2 = \delta(v, w)^2 - d_0^2$ and $\delta(u, w_0)^2 = \delta(u, w)^2 - d_0^2$. Since $\delta(v, w) > \delta(u, w)$, it follows that $\delta(v, w_0)^2 > \delta(u, w_0)^2$ and thus $\delta(v, w_0) > \delta(u, w_0)$. This means that $\delta(v, w') = \delta(v, w_0) + \delta(w_0, w')> \delta(u, w_0) +  \delta(w_0, w') \geq \delta(u, w')$, where the last inequality is due to triangle inequality. Hence, $\delta(u, w') < \delta(v, w')$.

It remains to prove that $\delta(u, v) < \delta(v, w')$. Since $\delta(v, w) > \delta(u, v)$, it is enough to prove that $\delta(v, w') \geq \delta(v, w)$. Let $z$ be the mid-point of the line segment $vw'$. Let $h = \delta(q, z)$, $a = \delta(v, z)$, $b = \delta(z, w_0)$, and $c = \delta(q, w_0)$. Note that $b < a$ and $b^2 + h^2 = c^2$. With these new notations, we can now write $\delta(v, w')^2 = (2a)^2 = 4a^2$ and 
\begin{align*}
    \delta(v, w)^2 = & \delta(v, w_0)^2 + \delta(w, w_0)^2
    \leq (a + b)^2 + d_0^2 = a^2 + 2ab + b^2 + d_0^2.
\end{align*}

Since $b < a$ and $b^2 + h^2 = c^2$, it follows that
\begin{align*}
    \delta(v, w)^2 < & a^2 + 2a^2 + (c^2 - h^2)  + d_0^2
                   = 4a^2 + (c^2 + d_0^2) - (a^2 + h^2).
\end{align*}
Note that $c^2 + d_0^2 = \delta(w, q)^2 < \delta(v, q)^2 = r^2$, since $w \in B = B_r(q)$. Also note that $a^2 + h^2 = \delta(v, q)^2 = r^2$. It follows that
$$\delta(v, w)^2 < 4a^2 + r^2 - r^2 = 4a^2 = \delta(v, w')^2.$$
This shows that $\delta(v, w) < \delta(v, w')$ and completes the proof.
\end{proof}

We formalize notation for angles here: For any three points $x, y, z$ in a Euclidean plane, we use $\angle yxz$ to denote the magnitude of the rotation from the ray $\overrightarrow{xy}$ to the ray $\overrightarrow{xz}$. This means $\angle yxz \neq \angle zxy$ unless $x, y, z$ are collinear. We note that in the rest of the proof, the sign of an angle does not matter if we only care about the cosine value of it, but it does matter or help with clarifying the situation sometimes.

\begin{lemma} \label{lem_angle}
Let $w \in \partial(B) \cap P$ and $\alpha = \angle wvu \in [-\pi, \pi]$. The following statements are true.\\
(1) If $|\alpha| \in [0, \frac{\pi}{3}]$, then $\delta(v, w) > \delta(v, u)$ if and only if $u \in lune(v, w)$.\\
(2) If $|\alpha| \in [\frac{\pi}{3}, \pi]$, then $\delta(v, w) > \delta(w, u)$ if and only if $u \in lune(v, w)$.
\end{lemma}

\begin{proof} [Proof of Lemma~\ref{lem_angle}]
Note that $u \in lune(v, w)$ if and only if $\delta(v, w) > \max \{\delta(v, u), \delta(w, u)\}$. Therefore, it suffices to show that $\delta(v, w) > \delta(v, u)$ implies $\delta(v, w) > \delta(w, u)$ if $|\alpha| \in [0, \frac{\pi}{3}]$ to prove (1), and that $\delta(v, w) > \delta(w, u)$ implies $\delta(v, w) > \delta(v, u)$ if $|\alpha| \in [\frac{\pi}{3}, \pi]$ to prove (2).

Recall that $d = \delta(v, u)$. Let $t = \delta(v, w)$ for convenience. By the law of cosines,
\begin{align*}
\delta(w, u)^2 =& \delta(v, u)^2 + \delta(v, w)^2 - 2 \delta(v, u) \delta(v, w) \cos \alpha = d^2 + t^2 - 2dt \cos \alpha.
\end{align*}

To prove (1), assume that $|\alpha| \in [0, \frac{\pi}{3}]$ and $\delta(v, w) > \delta(v, u)$. Then, $|\alpha| \in [0, \frac{\pi}{3}]$ implies that $\cos \alpha \geq \frac{1}{2}$ and $\delta(v, w) > \delta(v, u)$ means that $t > d$. Using the law of cosines, it follows that
\begin{align*}
\delta(v, w)^2 - \delta(u, w)^2 = & 2dt \cos \alpha - d^2
= d( (2 \cos \alpha) t - d)
\geq d(2 \cdot \frac{1}{2} t - d)
= d(t - d) > 0.
\end{align*}

Therefore, $\delta(v, w)^2 - \delta(u, w)^2 > 0$, meaning $\delta(v, w) > \delta(w, u)$.

To prove (2), assume that $|\alpha| \in [\frac{\pi}{3}, \pi]$ and $\delta(v, w) > \delta(w, u)$. Now, $|\alpha| \in [\frac{\pi}{3}, \pi]$ implies $\cos \alpha \leq \frac{1}{2}$, and $\delta(v, w) > \delta(w, u)$ implies $\delta(v, w)^2 - \delta(w, u)^2 > 0$. By the law of cosines, this means $d(2t \cos \alpha - d) = 2 d t \cos \alpha - d^2 = \delta(v, w)^2 - \delta(v, u)^2 > 0$. Since $d > 0$, we have $2t \cos \alpha - d > 0$. Since $\cos \alpha \leq \frac{1}{2}$, it follows that
$$\delta(v, w) - \delta(v, u) = t - d = 2 \cdot \frac{1}{2} t - d \geq 2 t \cos \alpha - d > 0.$$
This shows $\delta(v, w) > \delta(v, u)$ and hence completes the proof of Lemma~\ref{lem_angle}.
\end{proof}

\begin{lemma} \label{lem_bound}
The following statements are true.\\
\indent (1) There exists $w \in P \cap \partial(B)$ such that $\angle wvu \in [-\frac{\pi}{3}, \frac{\pi}{3}]$ and $u \in lune(v, w)$ if and only if $\theta \in [0, \frac{5}{6} \pi]$ and $d < r \cdot g(\theta)$ where
\begin{align*}
g(\theta) = 
\begin{cases}
2,  & \text{if   } \theta\in [0, \frac{\pi}{3}] \\
2 \cos(\theta - \pi/3), & \text{if } \theta \in [\frac{\pi}{3},  \frac{5}{6}\pi]\\
\text{undefined}, & \text{if } \theta \in (\frac{5}{6}\pi, \pi]
\end{cases}
\end{align*}

(2) There exists $w \in P \cap \partial(B)$ such that $\angle wvu \in [-\pi, -\frac{\pi}{3}] \cup [\frac{\pi}{3}, \pi]$ and $u \in lune(v, w)$ if and only if $d < r \cdot h(\theta)$ where
\begin{align*}
h(\theta) = 
\begin{cases}
2 \cos(\theta - \pi/3),  & \text{if   } \theta \in [0, \frac{2}{3} \pi] \\
2 (\cos \theta + 1), & \text{if } \theta \in [\frac{2}{3} \pi, \pi]
\end{cases}
\end{align*}
\end{lemma}

\begin{proof}
For convenience, for every $w \in P \cap \partial(B)$, let $\alpha_w = \angle wvu \in [-\pi, \pi]$. By the law of cosines, we have $\delta(v, u) = d$, $\delta(v, w) = 2r \cos(\theta + \alpha_w)$, and $\delta(u, w) = \sqrt{d^2 + 4r^2 \cos^2(\theta + \alpha_w) - 2d (2 r \cos(\theta + \alpha_w)) \cos \alpha_w}$. For every $\beta \in [0, \pi]$, let $A_\beta = \{\alpha_w \in [-\pi, \pi]: \cos(\beta + \alpha_w) \geq 0\}$. Note that for every $w \in P \cap \partial(B)$, $\theta + \alpha_w = \angle wvq \in [-\pi/2, \pi/2]$, meaning that $\cos(\theta + \alpha_w) \geq 0$ and therefore $\alpha_w\in A_\theta$.

\bigskip
We first prove (1). By Lemma~\ref{lem_angle}, we have
\begin{align*}
& \exists w \in P \cap \partial(B) \text{ such that } \alpha_w \in [-\frac{\pi}{3}, \frac{\pi}{3}] \text{ and } u \in lune(v, w)\\
\Leftrightarrow & \exists w \in P \cap \partial(B) \text{ such that } \alpha_w \in [-\frac{\pi}{3}, \frac{\pi}{3}] \text{ and } \delta(v, w) > \delta(v, u)\\
\Leftrightarrow & \exists \alpha \in [-\frac{\pi}{3}, \frac{\pi}{3}] \cap A_\theta \text{ such that } 2r \cos(\theta + \alpha) > d\\
\Leftrightarrow & d < \sup_{\alpha \in [-\frac{\pi}{3}, \frac{\pi}{3}] \cap A_\theta } \{2r \cos(\theta + \alpha)\} = 2r \sup_{\alpha \in [-\frac{\pi}{3}, \frac{\pi}{3}] \cap A_\theta} \{\cos(\theta + \alpha)\}.
\end{align*}

Recall that $\theta \in [0, \pi]$. If $\theta \in [0, \pi/3]$, then $\alpha^* = \arg \max_{\alpha \in [-\pi/3, \pi/3]} \{\cos (\theta + \alpha)\} = -\theta$ and $\sup_{\alpha \in [-\pi/3, \pi/3]} \{\cos (\theta + \alpha)\} = \cos (\theta + \alpha^*) = 0$. If $\theta \in [\pi/3, 5 \pi/6]$, then $\alpha^* = \arg \max_{\alpha \in [-\pi/3, \pi/3]} \{\cos (\theta + \alpha)\} = -\pi/3$ and $\sup_{\alpha \in [-\pi/3, \pi/3]} \{\cos (\theta + \alpha)\} = \cos (\theta + \alpha^*) = \cos(\theta - \pi/3) \geq 0$. If $\theta \in (5\pi/6, \pi)$, then $\cos (\theta + \alpha) < 0$ for all $\alpha \in [-\pi/3, \pi/3]$, meaning that $[-\pi/3, \pi/3] \cap A_\theta = \emptyset$. Hence, we conclude that
\begin{align*}
& \sup_{\alpha \in [-\frac{\pi}{3}, \frac{\pi}{3}] \cap A_\theta} \{\cos(\theta + \alpha)\} =
\begin{cases}
\cos(0) = 1 & \text{ if } \theta \in [0, \pi/3]\\
\cos(\theta - \pi/3) & \text{ if } \theta \in [\pi/3, 5\pi/6]\\
\text{undefined} & \text{ if } \theta \in (5\pi/6, \pi]
\end{cases}
\end{align*}
This shows that $g(\theta) = 2 \sup_{\alpha \in [-\frac{\pi}{3}, \frac{\pi}{3}] \cap A_\theta} \{\cos(\theta + \alpha)\}$ and therefore proves (1).

\bigskip
We now prove (2). By Lemma~\ref{lem_angle}, we have
\begin{align*}
& \exists w \in P \cap \partial(B) \text{ such that } \alpha_w \in [-\pi, -\frac{\pi}{3}] \cup [\frac{\pi}{3}, \pi] \text{ and} u \in lune(v, w)\\
\Leftrightarrow & \exists w \in P \cap \partial(B) \text{ such that } \alpha_w \in [-\pi, -\frac{\pi}{3}] \cup [\frac{\pi}{3}, \pi] \text{ and} \delta(v, w) > \delta(w, u)\\
\Leftrightarrow & \exists \alpha \in ([-\pi, -\frac{\pi}{3}] \cup [\frac{\pi}{3}, \pi]) \cap A_\theta \text{ such that }\\
& (2r \cos(\theta + \alpha))^2 > d^2 + 4r^2 \cos^2(\theta + \alpha) - 2d (2 r \cos(\theta + \alpha)) \cos \alpha\\
\Leftrightarrow & \exists \alpha \in ([-\pi, -\frac{\pi}{3}] \cup [\frac{\pi}{3}, \pi]) \cap A_\theta \text{ such that } d < 4r \cos(\theta + \alpha) \cos \alpha
\end{align*}

Note that $2\cos(\theta + \alpha) \cos \alpha = \cos \theta + \cos(\theta + 2 \alpha)$. It follows that there exists some $w \in P \cap \partial(B)$ such that $\alpha_w \in [-\pi, -\frac{\pi}{3}] \cup [\frac{\pi}{3}, \pi]$ and $u \in lune(v, u)$ if and only if there exists $\alpha \in ([-\pi, -\frac{\pi}{3}] \cup [\frac{\pi}{3}, \pi]) \cap A_\theta$ such that
$$d < 2r (\cos \theta + \sup_{\alpha \in ([-\pi, -\frac{\pi}{3}] \cup [\frac{\pi}{3}, \pi]) \cap A_\theta} \{\cos (\theta + 2 \alpha)\}).$$

For convenience, define
$$s(\theta) = \sup_{\alpha \in ([-\pi, -\frac{\pi}{3}] \cup [\frac{\pi}{3}, \pi]) \cap A_\theta} \{\cos (\theta + 2 \alpha)\}, \forall \theta \in [0, \pi].$$
We now make the following claim.

\begin{claim} \label{claim_sup}
For $\theta \in [0, \pi]$,
\begin{align*}
s(\theta) = 
\begin{cases}
\cos (\theta - \frac{2}{3}\pi)   &\text{if } \theta \in [0, \frac{2}{3}\pi]\\
1                                   &\text{if } \theta \in [\frac{2}{3}\pi, \pi]
\end{cases}
\end{align*}
\end{claim}

Before proving the claim, notice that if Claim~\ref{claim_sup} is true, then $2r(\cos \theta + s(\theta)) = 2r (\cos \theta + \cos (\theta - \frac{2}{3} \pi) = 2r \cos (\theta - \frac{1}{3}\pi)$ if $\theta \in [0, \frac{2}{3} \pi]$. This means that if Claim~\ref{claim_sup} is true, then
\begin{align*}
2r(\cos \theta + s(\theta)) =
\begin{cases}
2r \cos (\theta - \frac{1}{3}\pi), & \text{if } \theta \in [0, \frac{2}{3} \pi]\\
2r (\cos \theta + 1), & \text{if } \theta \in [\frac{2}{3}\pi, \pi]
\end{cases}
\end{align*}
It follows that $r \cdot h(\theta) = 2r(\cos \theta + s(\theta))$ for all $\theta \in [0, \pi]$, and this then completes the proof of (2). Therefore, to prove (2), it suffices to prove Claim~\ref{claim_sup}.

\begin{proof}[Proof of Claim~\ref{claim_sup}]
We first prove the easy case that $s(\theta) = 1$ if $\frac{2}{3}\pi \leq \theta \leq \pi$. Since the cosine function is bounded above by 1, we know that $s(\theta) \leq 1$ for all $\theta$. For every $\theta \in [\frac{2}{3}\pi, \pi]$, let $\alpha_\theta = -\frac{1}{2}\theta$. This means that $\alpha_\theta \in [-\frac{1}{2} \pi, -\frac{1}{3}\pi]$ and $\theta + \alpha_\theta = \frac{1}{2} \theta \in [\frac{1}{3} \pi, \frac{1}{2} \pi]$ and therefore $\cos(\theta + \alpha_\theta) \geq 0$. It follows that $\theta_\alpha \in ([-\pi, -\frac{1}{3} \pi] \cup [\frac{1}{3} \pi, \pi]) \cap A_\theta$. Since $\cos(\theta + 2 \alpha_\theta) = \cos(0) = 1$, it follows that $s(\theta) = 1$ for all $\theta \in [\frac{2}{3}\pi, \pi]$.

\bigskip
It remains to prove $s(\theta) = \cos(\theta - \frac{2}{3} \pi)$ for all $\theta \in [0, \frac{2}{3} \pi]$. Let $D(\theta) = ([-\pi, -\frac{1}{3} \pi] \cup [\frac{1}{3} \pi, \pi]) \cap A_\theta$ for all $\theta \in [0, \frac{2}{3} \pi]$ for convenience. Observe that
\begin{align*}
D(\theta) = & \{\alpha \in [-\pi, -\frac{1}{3} \pi] \cup [\frac{1}{3} \pi, \pi] : \cos(\theta + \alpha) \geq 0\}\\
= & 
\begin{cases}
[-\frac{\pi}{2} - \theta, -\frac{\pi}{3}] \cup [\frac{\pi}{3}, \frac{\pi}{2} - \theta] & \text{if } \theta \in [0, \frac{1}{6}\pi]\\
[-\frac{\pi}{2} - \theta, -\frac{\pi}{3}] & \text{if } \theta \in (\frac{1}{6}\pi, \frac{1}{2}\pi]\\
[-\pi, -\frac{\pi}{3}] \cup [\frac{3}{2}\pi - \theta, \pi]  & \text{if } \theta \in (\frac{1}{2} \pi, \frac{2}{3}\pi]
\end{cases}
\end{align*}

We will show that $s(\theta) = \sup_{D(\theta)} \{\cos(\theta + 2 \alpha)\} = \cos(\theta - \frac{2}{3} \pi)$ in each one of the three cases separately: $\theta \in  [0, \frac{1}{6} \pi]$, $\theta \in (\frac{1}{6}\pi, \frac{1}{2}\pi]$, and $\theta \in (\frac{1}{2} \pi, \frac{2}{3}\pi]$.

\medskip
\textbf{Case 1:} If $\theta \in [0, \frac{1}{6} \pi]$ and $\alpha \in D(\theta) = [-\frac{\pi}{2} - \theta, -\frac{\pi}{3}] \cup [\frac{\pi}{3}, \frac{\pi}{2} - \theta]$, then
\begin{align*}
    2 \alpha \in & [-\pi - 2 \theta, -\frac{2}{3} \pi] \cup [\frac{2}{3} \pi, \pi - 2 \theta]
    = [\pi - 2 \theta, \frac{4}{3} \pi] \cup [\frac{2}{3} \pi, \pi - 2 \theta]
    = [\frac{2}{3} \pi, \frac{4}{3} \pi].
\end{align*}
Therefore,
$$s(\theta) = \sup_{\alpha \in D(\theta)} \{\cos (\theta + 2 \alpha)\}
= \sup_{\beta \in [\frac{2}{3} \pi, \frac{4}{3} \pi]} \{\cos(\theta + \beta)\}.$$

Since $\theta \in [0, \frac{1}{6} \pi]$, it follows that
$$s(\theta) = \sup_{\beta \in [\frac{2}{3} \pi, \frac{4}{3} \pi]} \{\cos(\theta + \beta)\} = \cos(\theta + \frac{4}{3} \pi) = \cos (\theta - \frac{2}{3} \pi).$$

\medskip
\textbf{Case 2:} If $\theta \in (\frac{1}{6}\pi, \frac{1}{2}\pi]$ and $\alpha \in D(\theta) = [-\frac{\pi}{2} - \theta, -\frac{\pi}{3}]$, then $2 \alpha \in  [-\pi - 2\theta, -\frac{2}{3} \pi] = [\pi - 2 \theta, \frac{4}{3} \pi]$. We can then write
$$s(\theta) = \sup_{\alpha \in D(\theta)} \{\cos (\theta + 2 \alpha)\} = \sup_{\beta \in [\pi - 2 \theta, \frac{4}{3} \pi]} \{\cos (\theta + \beta)\}.$$
Notice that since $\theta \in (\frac{1}{6}\pi, \frac{1}{2}\pi]$, $\cos(\theta + \beta) \leq 0$ for all $\beta \in [\pi - 2 \theta, \frac{3}{2} \pi - \theta]$, and that $\cos(\theta + \beta) \geq 0$ if $\beta \in [\frac{3}{2} \pi - \theta, \frac{4}{3} \pi]$. It follows that

\begin{align*}
    s(\theta) & = \sup_{\beta \in [\frac{3}{2} \pi - \theta, \frac{4}{3} \pi]} \{\cos (\theta + \beta)\} = \cos(\theta + \frac{4}{3} \pi) = \cos(\theta - \frac{2}{3} \pi).
\end{align*}

\medskip
\textbf{Case 3:} Finally, if $\theta \in (\frac{1}{2} \pi, \frac{2}{3}\pi]$ and $\alpha \in D(\theta) = [-\pi, -\frac{\pi}{3}] \cup [\frac{3}{2}\pi - \theta, \pi]$, then $2 \alpha \in [-2 \pi, -\frac{2}{3}\pi] \cup [3\pi - 2\theta, 2\pi] = [0, \frac{4}{3}\pi] \cup [\pi - 2 \theta, 0]$. Therefore,
$$s(\theta) = \sup_{\beta \in [0, \frac{4}{3}\pi] \cup [\pi - 2 \theta, 0]} \{ \cos (\theta + \beta) \}.$$
Notice that since $\theta \in (\frac{1}{2} \pi, \frac{2}{3}\pi]$, we have
$$\sup_{\beta \in [0, \frac{4}{3}\pi]} \{ \cos (\theta + \beta) \} = \cos(\theta + \frac{4}{3} \pi) = \cos(\theta - \frac{2}{3} \pi),$$
and
$$\sup_{\beta \in [\pi - 2 \theta, 0]} \{ \cos (\theta + \beta) \} = \cos ( \theta + (\pi - 2 \theta) ) = \cos (\pi - \theta).$$
Therefore,
\begin{align*}
s(\theta) =& \sup_{\beta \in [0, \frac{4}{3}\pi] \cup [\pi - 2 \theta, 0]} \{ \cos (\theta + \beta) \}
= \max \{\cos(\theta - \frac{2}{3} \pi), \cos (\pi - \theta)\}
= \max \{\cos(\theta - \frac{2}{3} \pi), \cos (\theta - \pi)\}.
\end{align*}
Observe that $\cos(\theta - \frac{2}{3} \pi) > \cos (\theta - \pi)\}$ for all $\theta \in (\frac{1}{2} \pi, \frac{2}{3}\pi]$. It follows that $s(\theta) = \cos(\theta - \frac{2}{3} \pi)$ for all $\theta \in (\frac{1}{2} \pi, \frac{2}{3}\pi]$.

Overall, we proved that $s(\theta) = \cos(\theta - \frac{2}{3} \pi)$ for all $\theta \in [0, \frac{2}{3} \pi]$ and $s(\theta) = 1$ for all $\theta \in [\frac{2}{3} \pi, \pi]$, and therefore we completed the proof of Claim~\ref{claim_sup}.
\end{proof}
\end{proof}

Finally, we use Lemma~\ref{lem_redu_dim} and Lemma~\ref{lem_bound} to complete the proof of Lemma~\ref{lem_search_critiria}. Let functions $g, h$ be as defined in Lemma~\ref{lem_bound}. By Lemma~\ref{lem_bound}, we have equivalent statements as follows.
\begin{align*}
& \exists w \in P \cap \partial(B) \text{ such that } u \in lune(v, w)\\
\Leftrightarrow & \text{either }\exists w \in P \cap \partial(B) \text{ such that } \alpha_w \in [-\frac{\pi}{3}, \frac{\pi}{3}] \text{ and } u \in lune(v, w),\\
& \text{ or that } \exists w \in P \cap \partial(B) \text{ such that } \alpha_w \in [-\pi, -\frac{\pi}{3}] \cup [\frac{\pi}{3}, \pi] \text{ and } u \in lune(v, w)\\
\Leftrightarrow & d < r \cdot g(\theta) \text{ or } d < r \cdot h(\theta)\\
\Leftrightarrow & d < r \cdot \max \{g(\theta), h(\theta)\}
\end{align*}
By Lemma~\ref{lem_redu_dim}, it follows that there exists some $w \in B$ such that $u \in lune(v, w)$ if and only if $d < r \cdot \max\{g(\theta), h(\theta)\}$. Therefore, it suffices to show that $f(\theta) = \max \{g(\theta), h(\theta)\}$ for all $\theta \in [0, \pi]$.

By Lemma~\ref{lem_bound},
\begin{align*}
\max \{g(\theta), h(\theta)\} =
\begin{cases}
\max \{2, 2 \cos (\theta - \frac{1}{3} \pi)\}     & \text{if } \theta \in [0, \frac{1}{3} \pi]\\
2 \cos (\theta - \frac{1}{3} \pi)                  & \text{if } \theta \in [\frac{1}{3} \pi, \frac{2}{3} \pi]\\
\max \{2 \cos (\theta - \frac{1}{3} \pi), 2(\cos \theta + 1)\} & \text{if } \theta \in [\frac{2}{3} \pi, \frac{5}{6} \pi]\\
2(\cos \theta + 1) & \text{if } \theta \in [\frac{5}{6} \pi, \pi]
\end{cases}
\end{align*}
Note that $\cos (\theta - \frac{1}{3} \pi) \leq 1$ for all $\theta \in [0, \frac{1}{3} \pi]$, so $\max \{g(\theta), h(\theta)\} = \max \{2, 2 \cos (\theta - \frac{1}{3} \pi)\} = 2$ if $\theta \in [0, \frac{1}{3} \pi]$. Also note that for all $\theta$,

\begin{align*}
& (\cos \theta + 1) - \cos (\theta - \frac{1}{3} \pi) = 1 - 2 \sin (\frac{1}{2}(2 \theta - \pi/3)) \sin(\frac{1}{2}(\pi/3))
= 1 - \sin (\theta - \pi/6) \geq 0.
\end{align*}
Therefore, if $\theta \in [\frac{2}{3} \pi, \frac{5}{6} \pi]$, then
\begin{align*}
    \max \{g(\theta), h(\theta)\} & = \max \{2 \cos (\theta - \frac{1}{3} \pi), 2(\cos \theta + 1)\} = 2(\cos \theta + 1).
\end{align*}
It follows that $\max \{g(\theta), h(\theta)\} = f(\theta)$ for all $\theta \in [0, \pi]$.
\end{proof}

\subsection{Future directions for implementing conflicting nodes}

Although conflicting nodes seem to be a great tool to \emph{get out of local minimum points}, it is not so clear how to utilize them to run the search efficiently in practice. We now discuss some challenges on implementation of conflicting nodes. These challenges lead to several interesting future directions on this subject.

The first challenge is the fact that a non-neighbor of the local minimum node $
v$ could be a conflicting node to multiple neighbors of $v$: Recall that $w$ is a conflicting node for the edge $\overrightarrow{vu}$ in an MRNG if $u \in lune(v, w)$. By this definition, it is possible that more than one neighbor of $v$ are all contained in $lune(v, w)$, and in that case $w$ would be a conflicting node to multiple edges coming out of $v$. It is then inefficient to store the duplicated conflicting nodes when building the graph and to check them during the search. For each non-neighbor $w$ of $v$, assigning it as a conflicting node to different edges coming out of $v$ may result in different search performances, but we do not know which assignment optimizes the search performance. Therefore, one future direction would be to figure out how to assign non-neighbors to the neighbors of a node as conflicting nodes so that the search performance can be optimized.

Furthermore, for every node $v$ in an MRNG and for every non-neighbor $w$ of $v$, let $k_v(w)$ denote the number of edges coming out of $v$ that $w$ is conflicting to. We observed that for a fixed number of data points, $k_v(w)$ grows as the dimension of the dataset increases, as shown in Figure~\ref{fig_conf_dup_count}. This agrees with the curse of dimensionality and tells us that in high dimension, it would be more crucial to figure out how to assign the conflicting nodes to the neighbors. Future work on this assignment might involve different algorithms for low and high dimensions.

\begin{figure}[H]
\centering
\subfigure{\includegraphics[scale = 0.49]{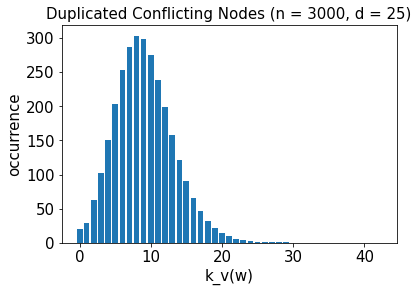}}%
\subfigure{\includegraphics[scale = 0.49]{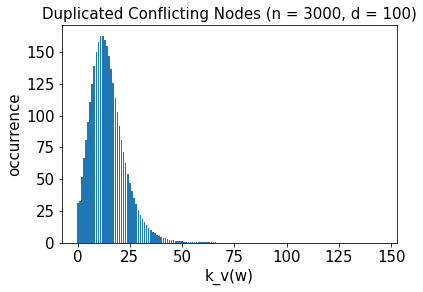}}%
\caption{Histograms of $k_v(w)$ among all pairs of $(v, w)$, where $v, w \in V(G)$ and $w$ is a conflicting node to some edge coming out of $v$: on the left, the dataset has 3000 data points in 25-dimension; on the right, the dataset has 3000 data points in 100-dimension.}
\label{fig_conf_dup_count}
\end{figure}

In addition, let a normal greedy search algorithm be a phase I search that returns a local minimum point $v$ with respect to the query $q$, and let the conflict search starting from $v$ be phase II search. Another future direction would be to study the trade-off between phase I and phase II search and how to distribute computation power between the two phases to obtain the highest search accuracy. In particular, as we would normally encounter several local minimum points during phase I search, when is the moment to stop phase I and move on to phase II?

Last but not least, it is possible that phase I search converges to the true nearest neighbor of a query faster than phase II search. In that case, it may make more sense to only use phase II search to find a conflicting $w$ to an edge $\overrightarrow{vu}$ for some $u \in N_{out}(v)$ that is closer to $q$ than $v$ and then get back to phase I search, every time when phase I is stuck at a local minimum node with respect to $q$. In other words, it may make sense to only use phase II search as a subroutine to get around local minimum points during phase I search. There are many possible ways to make use of conflicting nodes, and thus in future work there needs to be more experiments or theory to find out which one is the best.